\documentclass[11pt,runningheads]{llncs}

\usepackage[T1]{fontenc}
\usepackage{amsmath,amssymb,mathtools}
\usepackage[margin=1in]{geometry}
\usepackage{graphicx}
\usepackage{url}
\usepackage{microtype}

\newcommand{\E}{\mathbb{E}}

\begin{document}

\title{Uniform-Price Mechanism Design for Automated Market Makers}
\titlerunning{Uniform-Price Mechanism Design for AMMs}
\subtitle{Full Version}

\providecommand{\fullversionnote}{}
\author{Zhike Peng\thanks{Corresponding author. The author thanks Jing Chen
for her guidance and support, and Jichen Li, Xuanzhi Xia, \mbox{Wentao Zhou},
and the anonymous reviewers for helpful comments. This work is partially supported by the Beijing
Advanced Innovation Center for Future Blockchain and Privacy Computing.\fullversionnote}\orcidID{0009-0003-1974-2012}}
\authorrunning{Z. Peng}
\institute{Department of Computer Science and Technology,\\
Tsinghua University, 30 Shuangqing Road, Beijing 100084, China\\
\email{pzk22@mails.tsinghua.edu.cn}}

\maketitle

\begin{abstract}
Automated market makers (AMMs) are a central source of liquidity in
decentralized exchanges, but their reserve curves make mechanism design
substantially different from classical double auctions.

We study a batch double auction with buyers, sellers, and a
liquidity-nondecreasing AMM: post-trade reserves must have weakly larger
product than the initial reserves.
Each user has two-dimensional private information: a unit value and a
demand or supply quantity. The mechanism must clear user trades at one
uniform price while keeping the post-trade AMM state feasible. The objective is
user-side trading volume, benchmarked against the optimal feasible
individually rational~(IR) uniform-price allocation.

We give a randomized uniform-price mechanism based on binary quantity buckets.
It randomly chooses one of three clearing modes: peer-to-peer matching, AMM
sales to buyers, and seller sales to the AMM\@. Every deterministic realization
of the mechanism is ex-post feasible and satisfies IR and weak
quantity incentive compatibility: conditional on any quantity report,
truthful unit-value reporting is dominant whenever it gives a non-negative utility, no unit-value deviation gives a positive utility if the truthful one doesn't, and users cannot benefit by under-reporting the quantity together with their
truthful unit value. Let $V^\ast$ denote the optimal feasible IR uniform-price
volume without additional incentive constraints, and let
$\gamma=\bar q/\underline q$, where $\bar q$ and $\underline q$ are the upper
and lower bounds on users' quantity domain. For $\gamma\ge1$, the mechanism
achieves an expected trading volume of
\[
  \Omega(1/(1+\log\gamma)^2)\,V^\ast - O(\bar q).
\]

We also show that full incentive compatibility over the two-dimensional type space forces a worst-case
approximation ratio of $O(1/\gamma)$ at every prescribed initial reserve ratio,
even when both AMM reserves are scaled by an arbitrarily large common factor and
the optimal volume~$V^\ast$ is arbitrarily large relative to $\bar q$.

\end{abstract}

\newpage

\section{Introduction}

Public blockchains have made financial applications programmable and
continuously settled on a shared ledger. Decentralized exchanges are among the
most important examples, and automated market makers (AMMs) are their main source of on-chain liquidity:
constant-function automated market makers (CFAMMs) have been reported to handle daily trading volumes between
USD 3 billion and USD 6 billion~\cite{canidio2023batching}, and they are now standard
infrastructure in decentralized finance~\cite{mohan2022defi}. Instead of
maintaining an order book, an AMM replaces counterparty search with a
deterministic trading rule maintained by a liquidity pool. Early theory mostly
focused on CFAMMs, especially
constant-product automated market makers (CPAMMs), and characterizes the
geometry of these pools, including their prices, slippage, and arbitrage
behavior~\cite{angeris2020cfmm,angeris2021uniswap}.

The blockchain execution environment also shapes the relevant market design
problem. Transactions are ordered and executed block by block. On a
decentralized exchange, users broadcast transactions and validators or builders
select an execution ordering within each block; for an AMM, that ordering
determines the sequence of reserve states and hence the prices faced by
contemporaneous traders~\cite{ferreira2023vsr}. Sequential execution therefore
creates ordering externalities, arbitrage races, and sandwich
attacks~\cite{daian2020flashboys,qin2022quantifying}. Batch execution is a
natural design response. By collecting orders over a clearing interval and
executing them at one uniform price, a batch mechanism treats contemporaneous
orders symmetrically and reduces the value of manipulating transaction order.
This follows the same broad market-design logic as frequent batch auctions,
which replace continuous serial processing with discrete-time uniform-price
clearing~\cite{budish2015hft}.

Prior work on batch AMMs shows that batching can be combined with reserve-state
rules that differ from the usual constant-function execution rule. Canidio and
Fritsch~\cite{canidio2023batching} propose the function-maximizing AMM (FM-AMM),
in which batched trades execute at the post-batch marginal price and, in the
constant-product case, the pool may move to a higher product level. Ramseyer et
al.~\cite{ramseyer2023augmenting} study CFAMMs in batch exchanges under a
nondecreasing trading-function requirement. Motivated by this reserve-state
feature, we abstract away from the AMM response rule itself and impose only
that, after batch clearing, the constant-product AMM's reserve product is
weakly larger than before. This abstraction gives the batch mechanism
additional flexibility in clearing order imbalances while preserving a minimal
liquidity-nondecreasing requirement on the AMM\@. We call an AMM satisfying
this reserve feasibility constraint a liquidity-nondecreasing AMM (LND-AMM).

This leads to a different question from the existing batch-AMM literature.
Buyers and sellers submit both a unit value and a quantity. A mechanism
chooses user allocations and payments at one uniform price, while the AMM
absorbs the residual net imbalance induced by those allocations subject to the
LND feasibility constraint. Inspired by maximal-volume matching in the
double-auction literature~\cite{niu2013maximizing}, and by the importance of
executed order flow in AMM markets~\cite{canidio2023batching,adams2024amamm},
we measure efficiency by total user-side trading volume. Thus a peer-to-peer
trade of $q$ units contributes $2q$, while a user-AMM trade of $q$ units
contributes $q$. Accordingly, the objective counts executed user quantities, including
orders executed against AMM liquidity, without assigning the AMM an external
value for its ending inventory; we discuss its relation to welfare and gains
from trade below.
Uniform pricing alone does not remove the standard tension
between incentives and efficiency in double
auctions~\cite{myerson1983efficient,mcafee1992dominant}; here that tension is
compounded by the nonlinear reserve feasibility imposed by the AMM.

Two technical difficulties drive the study. First, each user has private
information in both a unit value and a quantity. This makes full
incentive compatibility~(IC) much more restrictive than in
single-parameter double auctions and forces us to distinguish value incentives
from quantity incentives. Second, the AMM constraint is global and nonlinear:
the feasibility of serving a set of users depends not only on who trades at the
uniform price, but also on the net imbalance absorbed by the AMM\@. The
benchmark is therefore not a local matching objective, but the best feasible
IR uniform-price batch allocation.

\paragraph{Contributions.}
First, we formulate a uniform-price double auction with an LND-AMM and
two-dimensional user types. The benchmark $V^\ast$ is the maximum user-side
volume among all feasible IR uniform-price allocations.

Second, we give a randomized uniform-price mechanism. It groups reported
quantities into binary buckets and randomly chooses one of three submechanisms:
$M_1$ clears peer-to-peer trades, $M_2$ lets the AMM sell to buyers, and $M_3$
lets sellers sell to the AMM\@. For a fixed random realization, clearing is
computed by sorting retained orders and scanning once. The most important and
technically difficult component is $M_1$, which is inspired by the
strongly-budget-balanced double auction of Segal-Halevi, Hassidim, and Aumann~\cite{segalhalevi2016sbba}. Their mechanism handles unit-demand and
unit-supply agents, whereas our bucket thresholds convert heterogeneous quantities into
unit-like blocks while preserving feasibility.

Third, we prove that every deterministic realization of the mechanism is
feasible and satisfies individual rationality (IR) and weak quantity incentive
compatibility (WQIC). Let
$\gamma=\bar q/\underline q$, where $\bar q$ and $\underline q$ are the upper
and lower bounds on users' quantity domain. For $\gamma\ge1$, the mechanism
achieves
\[
  \E[V_M] \ge
  \Omega(1/(1+\log\gamma)^2)V^\ast - O(\bar q).
\]
The quadratic logarithmic loss comes from the binary
quantity thresholds; the additive term comes from boundary orders, threshold
truncation, and trade reduction.

Finally, we show that some relaxation of full two-dimensional IC is
structurally necessary for a logarithmic guarantee in this framework. If a feasible, ex-post IR randomized uniform-price
mechanism is required to be fully IC over the complete two-dimensional type
space, even in the expected-utility sense rather than the ex-post sense, then for
every $\gamma>1$ and every market size $N$ there is an instance
with $N$ buyers and $N$ sellers, with $V^\ast=2N\bar q$, on which its expected
volume is at most $(3/\gamma)V^\ast$. This remains true at every prescribed
initial reserve ratio and under an arbitrarily large common scaling of the two
initial reserves. Thus the worst-case
approximation ratio achievable under full IC is $O(1/\gamma)$, and the
approximation-ratio separation is not an artifact of the additive $O(\bar q)$
term in the positive guarantee or of a nearly empty pool.

\section{Related Work}


The idea of automated market making has several important predecessors. In
prediction markets, Hanson~\cite{hanson2007lmsr} introduced the logarithmic
market scoring rule, which provides continuous liquidity through an automated
pricing rule. In the blockchain setting, Buterin~\cite{buterin2017path} discussed
path-independent on-chain market makers and the constant-product invariant,
and Adams~\cite{adams2018uniswap} crystallized this design in the early
Uniswap white paper. These developments led to a rich and rapidly growing
literature on AMMs, whose early theoretical focus was largely on CFAMMs,
especially CPAMMs. Angeris and Chitra~\cite{angeris2020cfmm} formulate
CFAMMs as invariant-based trading mechanisms and study their pricing and
oracle properties. Angeris et al.~\cite{angeris2021uniswap} analyze Uniswap-style
CPAMMs, clarifying how reserves, slippage, and arbitrage determine trading
outcomes. Bartoletti, Chiang, and Lluch-Lafuente~\cite{bartoletti2021theory}
give a broader formal theory of AMMs in DeFi, emphasizing the structure shared
across common invariant-based designs. Mohan~\cite{mohan2022defi} surveys the
broader decentralized-exchange ecosystem.

Later work shifts attention from static CFAMM or CPAMM geometry to market
microstructure and mechanism issues around AMMs. Milionis et
al.~\cite{milionis2022lvr} identify loss-versus-rebalancing as a central source
of liquidity-provider losses. Milionis, Moallemi, and
Roughgarden~\cite{milionis2024fees} study how fees reshape arbitrage profits
and AMM performance. Adams et al.~\cite{adams2024amamm} propose an auction-managed AMM
that auctions the right to set swap fees and collect fee revenue. In parallel, Daian
et al.~\cite{daian2020flashboys} and Qin et al.~\cite{qin2022quantifying} show
how transaction ordering on blockchains creates sandwich attacks, arbitrage
races, and broader extractable-value concerns. These works explain why batch
execution and auction design matter; and our study asks what can still be
guaranteed once one also imposes uniform pricing and strategic reporting by
both buyers and sellers.

The LND-AMM reserve constraint is closest to batch AMMs that let the pool move
to a weakly higher product level rather than staying on the old constant-product
curve. Canidio and Fritsch~\cite{canidio2023batching} introduce the FM-AMM,
motivated by the observation that batching can eliminate arbitrage profits,
loss-versus-rebalancing, and sandwich attacks when arbitrageurs compete inside
the batch. Their AMM executes all trades in a batch at the post-batch marginal
price and, for fixed prices, is equivalent to maximizing an AMM objective
function. In the constant-product case, this means moving to the highest
attainable product curve. Ramseyer et
al.~\cite{ramseyer2023augmenting} study the more general problem of
incorporating CFAMMs into batch exchanges. Their axiomatic framework includes
uniform prices, asset conservation, best-response execution of limit orders,
and the requirement that each CFAMM's trading function not decrease; they then
analyze equilibrium computation, price coherence, and tradeoffs among these
desiderata.

Our model takes only the nondecreasing reserve-function implication from this
line of work. Unlike Canidio and Fritsch~\cite{canidio2023batching}, we do not
prescribe how the AMM chooses its post-batch reserve point at a given price.
Unlike Ramseyer et al.~\cite{ramseyer2023augmenting}, we do not study equilibrium
computation for multiple CFAMMs and limit orders.
The AMM reserve condition enters here only as a feasibility constraint in a
strategic uniform-price double auction. The central question is how much
user-side volume can be guaranteed when buyers and sellers have two-dimensional
private information.

Other work explicitly treats AMMs as mechanisms. Ferreira and
Parkes~\cite{ferreira2023vsr} analyze decentralized exchange design under
verifiable sequencing rules. Chan, Wu, and Shi~\cite{chan2025amm} construct a
constant-potential batch swap mechanism that is arbitrage resilient in the
plain sequencing model and strategy-proof under weak fair sequencing. Their
goal is MEV resilience against strategies that may submit multiple orders;
they neither require batch-wide uniform pricing nor approximate a double-auction
objective, and leave welfare optimization open. Li, Shi, and
Zhang~\cite{li2026trilemma} work in that framework and prove a qualitative
trilemma: among IC, weak local efficiency~(wLE), and uniform pricing, any two
can be achieved but all three cannot. Here wLE is a local fulfillment condition
only for orders whose limits are no more stringent than the AMM's initial
marginal price.

Their trilemma and our full-IC upper bound do not imply one another. The
trilemma rules out the exact conjunction IC$+$wLE$+$uniform pricing but gives
no approximation rate to a global objective. Our theorem instead places a
quantitative $O(1/\gamma)$ upper bound on the worst-case volume guarantee of
any fully IC uniform-price mechanism, without assuming wLE; conversely, a poor
global volume ratio need not witness a wLE violation. Our model also permits
the reserve product to increase and covers randomized mechanisms under
expected-utility IC, whereas their stated trilemma uses deterministic
constant-potential mechanisms. Thus the two results identify complementary,
rather than logically nested, incentive-efficiency tensions.

Because both of our main results are stated relative to a global volume
benchmark, we next situate that objective in the AMM and double-auction
literatures. In AMM work, trading volume has often been used as an empirical
measure of market activity or as an input to fee-revenue calculations, rather
than as the formal mechanism-design objective. For example, Canidio and
Fritsch~\cite{canidio2023batching} use the scale of AMM trading volume to
motivate the market-design problem, while Adams et
al.~\cite{adams2024amamm} optimize revenue from uninformed order flow. In the
double-auction literature, Niu and Parsons~\cite{niu2013maximizing} explicitly
study maximal-volume matching---maximizing the amount of commodity traded---and
motivate it by transaction throughput and the interests of marketplaces whose
revenue grows with trading activity. Our benchmark adopts this order-execution
interpretation, while adding trades against AMM liquidity and strategic
quantity reports.

Classical double-auction theory provides the relevant incentive-efficiency
context. Myerson and Satterthwaite~\cite{myerson1983efficient} show that,
in the standard bilateral-trade model with independent private values and
overlapping supports, IR, IC, budget balance, and
ex-post allocative efficiency cannot all be satisfied; in the bilateral setting,
this exact allocative efficiency is equivalent to Pareto efficiency. Our full-IC upper
bound differs in both the market and the efficiency notion: the market contains
many buyers and sellers plus an AMM reserve constraint, and the target is not
exact surplus efficiency but approximation to the global feasible
uniform-price volume $V^\ast$. McAfee's trade-reduction
mechanism~\cite{mcafee1992dominant} is a central truthful double-auction benchmark. The
peer-to-peer part of our mechanism is more directly inspired by the
strongly-budget-balanced double auction of Segal-Halevi, Hassidim, and
Aumann~\cite{segalhalevi2016sbba}, whose unit-demand and unit-supply structure we
adapt through quantity buckets.

\section{The model}

The rest of the paper uses the following batch-clearing model. We first specify
the AMM reserve constraint, then the users' reports and utilities, and finally
the mechanism properties and benchmark.

There are two assets: a traded asset $X$ and a numeraire asset $Y$. The AMM
starts from positive reserves $P_0=(x_0,y_0)$ and a post-trade state is denoted
$P'=(x,y)$. Reserve states have non-negative coordinates, and we require
\[
  xy \ge x_0y_0.
\]
We call this a liquidity-nondecreasing AMM (LND-AMM). It weakly generalizes the
usual constant-product constraint by allowing the batch to end with strictly more
product-liquidity.

For $q<x_0$, define
\[
  C(q)=\frac{x_0y_0}{x_0-q}-y_0.
\]
If the AMM sells $q>0$ units of $X$, then $C(q)$ is the minimum required inflow
of $Y$. For an AMM purchase of $q>0$ units of $X$, the maximum outflow of $Y$
is $-C(-q)$. The average sell price $C(q)/q=y_0/(x_0-q)$ is increasing in $q$,
while the average buy price $-C(-q)/q=y_0/(x_0+q)$ is decreasing in $q$.

\paragraph{Users.}
Let $\mathcal{B}$ and $\mathcal{S}$ denote the buyer and seller sets,
respectively. Throughout, $i\in\mathcal{B}$ indexes buyers,
$j\in\mathcal{S}$ indexes sellers, and $k$ ranges over
$\mathcal{B}\cup\mathcal{S}$ to index a generic user.

Buyer $i$ reports $(\hat b_i,\hat q_i)$, where $\hat b_i$ is a bid
and $\hat q_i$ is a demand quantity. Seller~$j$ reports
$(\hat s_j,\hat q_j)$, where $\hat s_j$ is an ask and $\hat q_j$ is a
supply quantity. True and reported unit values and unit costs are non-negative.
All true and reported quantities lie in
$[\underline q,\bar q]$, where $0<\underline q\le\bar q<\infty$, and
\[
  \gamma=\bar q/\underline q.
\]
Ties in sorting and allocation are resolved by a fixed report-independent
priority rule. Prices, utilities, and reserve feasibility use the unperturbed
numerical values.

Buyer $i$ has true type $(b_i,q_i)$, where $b_i$ is the unit value and $q_i$ is
the true demand quantity. If the buyer receives $z_i$ units and pays
$t_i$, her utility is
\[
  u_i=b_i\min\{z_i,q_i\}-t_i.
\]
Seller $j$ has two-dimensional private type $(s_j,q_j)$, where $s_j$ is the
unit cost and $q_j$ is the true supply quantity. Each seller also has a
public, exogenous over-supply penalty factor $\lambda_j\ge 1$. If the seller
sells $z_j$ units and receives
$t_j$, her utility is
\[
  u_j=
  t_j - s_jz_j
  -(\lambda_j-1)s_j\max\{z_j-q_j,0\}.
\]
The public factors $\lambda_j$ may differ across sellers; larger values model
the extra cost of sourcing units beyond owned supply. This term has no buyer
counterpart because units beyond a buyer's demand can simply have zero marginal
value, as captured by $\min\{z_i,q_i\}$. In contrast, a seller allocated
above owned supply must procure or borrow the shortfall. For a unified
presentation, we take the factors to be public, so buyers and sellers both have
two-dimensional private types.

\paragraph{Deterministic mechanisms and feasibility.}
These preferences are evaluated under uniform-price allocations that may leave
a residual imbalance for the AMM\@.
A deterministic mechanism maps the reported orders and the public AMM state to
allocations $z_i$ for buyers and $z_j$ for sellers, and payments $t_i,t_j$.
It is uniform-price if it also outputs a price $\pi$ such that
$t_i=\pi z_i$ for buyers and $t_j=\pi z_j$ for sellers. Write
$Q_B=\sum_{i\in\mathcal{B}} z_i$ and $Q_S=\sum_{j\in\mathcal{S}} z_j$. The AMM absorbs the net
imbalance, so the induced state is
\[
  P'=
  \bigl(x_0-Q_B+Q_S,\,
  y_0+\pi(Q_B-Q_S)\bigr).
\]
The deterministic outcome is feasible if $0\le z_i\le \hat q_i$ for every
buyer, $0\le z_j\le \hat q_j$ for every seller, and $P'$ has non-negative
coordinates satisfying $xy\ge x_0y_0$.

To compare reports, fix $M$, a user $k$ with true type $(v,q)$, all other
users' reports, and the public parameters (including reserves and penalties).
Here $v$ is the buyer's unit value $b_i$ or the seller's unit cost $s_j$.
Write $u_k(\hat v,\hat q\mid v,q)$ for the buyer or seller utility defined
above, evaluated at the allocation and payment that $M$ produces when
$k$ reports $(\hat v,\hat q)$. The true type $(v,q)$ stays fixed; the
mechanism, other reports, and public parameters are suppressed in this notation.

IR means $u_k(v,q\mid v,q)\ge0$ for every user, true type, other-report
profile, and public parameter choice. Full IC requires
$u_k(v,q\mid v,q)\ge u_k(\hat v,\hat q\mid v,q)$ for every admissible
report $(\hat v,\hat q)$.

\paragraph{Weak quantity incentive compatibility (WQIC).}

\begin{definition}[WQIC]\label{def:wqic}
A deterministic mechanism $M$ is weak quantity incentive compatible if both
conditions hold for every
user, true type, other-report profile, and public parameter choice:
\begin{enumerate}
\item For every admissible report $(\hat v,\hat q)$,
\[
 u_k(\hat v,\hat q\mid v,q)\le\max\{u_k(v,\hat q\mid v,q),0\}.
\]
\item For every $\hat q\in[\underline q,q]$,
\[
 u_k(v,\hat q\mid v,q)\le u_k(v,q\mid v,q).
\]
\end{enumerate}
\end{definition}

Thus, once the quantity report is fixed, truthful unit-value reporting weakly maximizes
the user's utility among all unit-value reports whenever the truthful unit-value report
yields non-negative utility at that quantity. If the truthful unit-value report yields
non-positive utility at that quantity, then no unit-value report with the same quantity
can yield positive utility. In addition, a user cannot gain by understating demand or supply at the truthful unit value.

The distinction between the non-negative and non-positive cases in the first condition
is needed because a large quantity overstatement can make the truthful unit value undesirable.
Reporting a quantity far above the true quantity may cause the allocation to exceed the user's
true demand or supply, in which case the user may prefer to report a less competitive unit
value in order to avoid trade. The first condition therefore requires truthful unit-value
reporting to be optimal only when it yields non-negative utility, and otherwise requires
that no unit-value report at the same quantity can yield positive utility.

For a deterministic IR, WQIC mechanism, the definition also gives a useful reduction
in the set of reports that need to be considered. A deviation
$(\hat v,\hat q)$ giving non-positive utility is weakly dominated by
truthful reporting $(v,q)$, which gives non-negative utility by IR.
If instead $(\hat v,\hat q)$ gives positive utility, then
$(v,\hat q)$ must also give positive utility: otherwise the first WQIC
condition would rule out every positive-utility value report at
$\hat q$. Applying that condition again shows that $(v,\hat q)$
weakly dominates $(\hat v,\hat q)$. If $\hat q\le q$, the second
condition further gives dominance by $(v,q)$. Consequently, only
quantity overstatements at the truthful unit value remain as potentially
profitable deviations.

These conditions reflect the user-side-volume objective: quantity
withholding reduces the amount offered for execution, whereas larger
quantity reports leave room for additional trades.

\paragraph{Randomized mechanisms and ex-post properties.}
A randomized mechanism is a family of deterministic mechanisms
$\{M^\omega\}_{\omega\in\Omega}$ with a report-independent seed distribution.
It satisfies ex-post feasibility, IR, or full IC if every $M^\omega$
satisfies the corresponding deterministic property. Full IC in expected
utility requires truthful reporting to maximize expected utility over both
report dimensions. Ex-post full IC implies full IC in expected utility by
averaging the incentive inequalities.

For a randomized mechanism, ex-post WQIC means that the deterministic
definition holds for every $M^\omega$, using the same seed across the
reports being compared. The preceding reduction applies
separately to each realization. In particular, ex-post WQIC and ex-post IR
imply that, at the truthful quantity, truthful unit-value reporting weakly
maximizes realized utility among all unit-value reports with that quantity.

\paragraph{Benchmark.}
We separate the approximation target from these incentive constraints.
For an allocation, define user-side volume
\[
  V=\sum_{i\in\mathcal{B}} z_i+\sum_{j\in\mathcal{S}} z_j.
\]
Let $V^\ast$ be the maximum volume among all feasible IR uniform-price
allocations for the reported instance. This benchmark imposes no incentive
constraints beyond IR: it is the best volume that could be obtained by a
feasible IR uniform-price allocation if the reported types were taken as fixed
inputs and strategic misreporting were ignored. Equivalently, after all user
reports are fixed, we compare our mechanism and the benchmark on that same
reported instance. Thus $V^\ast$ is an upper bound on what one can hope to
achieve under additional incentive requirements such as WQIC or full IC\@. Partially served users are permitted in this benchmark.

Volume is deliberate rather than a proxy for welfare. It measures how much
submitted user demand and supply the batch executes, including orders cleared
against the AMM. A peer-to-peer trade of $q$ units executes two user orders and
therefore contributes $2q$, whereas a user-AMM trade of $q$ units executes one
user order and contributes $q$. For user-to-user trades, the gains-from-trade
objective instead weights the traded quantity by the buyer--seller surplus margin, so
neither objective generally approximates the other. Extending gains from trade
or social welfare to user-AMM trades would additionally require an external
valuation of changes in the pool's reserve inventory, which wasn't considered
in the AMM mechanism-design literature closest to our setting. Studying gains
from trade for AMMs is an interesting direction for future work. Our
results should therefore be read as user-side order-execution guarantees, not
as welfare guarantees. This choice is consistent with maximal-volume matching
in double auctions and with the operational role of executed order flow in AMM
markets. We do not claim that volume is universally preferable to gains from
trade or social welfare.

\section{Our mechanism}

With the model and benchmark established, we now introduce our randomized
mechanism $M$.

Let
\[
  L=\max\left\{1,\left\lceil \log_2(\bar q/\underline q)\right\rceil\right\}.
\]
For buyers, define buckets $I_0,\ldots,I_{L-1}$ by
\[
  I_l=\{i\in\mathcal{B}:\hat q_i\in[2^l\underline q,2^{l+1}\underline q)\},
  \quad l<L-1,
\]
and
\[
  I_{L-1}=\{i\in\mathcal{B}:\hat q_i\in[2^{L-1}\underline q,\bar q]\}.
\]
Seller buckets $J_l$ are defined analogously. The full mechanism $M$ samples
$\mu$ uniformly from $\{1,2,3\}$ and runs $M_\mu$. Thus, in this mechanism, the
random realization $\omega$ consists only of the selected mode $\mu$ and the
selected bucket threshold or thresholds used by that mode.

The three modes correspond to the three ways in which the optimal uniform-price
allocation can create volume: $M_1$ handles peer-to-peer user matching, $M_2$
handles the case in which the AMM sells net quantity to buyers, and $M_3$
handles the case in which sellers sell net quantity to the AMM.

The central and most technically difficult case is $M_1$. It is inspired by the
strongly-budget-balanced double auction of Segal-Halevi, Hassidim, and Aumann~\cite{segalhalevi2016sbba}. Their mechanism is designed for unit-demand buyers
and unit-supply sellers, where the allocation problem is essentially
single-parameter and IC can be handled by standard critical-price arguments.
Here, by contrast, users have heterogeneous quantities and two-dimensional
types. The bucket thresholds in $M_1$ convert reported quantities into
unit-like blocks, allowing a related trade-reduction idea to be applied while
preserving feasibility and the WQIC guarantee.
Appendix~\ref{app:m1-example} gives a numerical example of its clearing index,
boundary ask, and both clearing branches.

\paragraph{$M_1$: Peer-to-peer clearing.}
Sample buyer and seller thresholds $l_B,l_S$ independently and uniformly from
$\{0,\ldots,L-1\}$. Let
\[
  \Delta_B=2^{l_B}\underline q,
  \qquad
  \Delta_S=2^{l_S}\underline q.
\]
Retain buyers in $\cup_{l\ge l_B}I_l$ and sort them by decreasing bid
$\hat b_{(1)}\ge\cdots\ge\hat b_{(n)}$. Retain sellers in
$\cup_{l\ge l_S}J_l$ and sort them by increasing ask
$\hat s_{(1)}\le\cdots\le\hat s_{(m)}$.
Assume $n,m>0$; otherwise this case is trivial.

Suppose first that $\Delta_B\ge \Delta_S$ and set
$r=\Delta_B/\Delta_S$. One buyer and $r$ sellers form one candidate trade set.
The effective seller ask of the $h$th trade set is
\[
  a_h=\hat s_{(r(h-1)+1)}.
\]
Let
\[
  \kappa=
  \max\bigl(\{0\}\cup
  \{h:1\le h\le n,\ r(h-1)+1\le m,\ a_h\le \hat b_{(h)}\}\bigr).
\]
If $\kappa=0$, no trade occurs. Otherwise define
\[
  \alpha_\kappa=
  \begin{cases}
    \hat s_{(r\kappa+1)}, & r\kappa+1\le m,\\
    +\infty, & r\kappa+1>m.
  \end{cases}
\]
If $\alpha_\kappa\le \hat b_{(\kappa)}$, set the price
$p^\ast=\alpha_\kappa$ and trade the first $\kappa$ buyers and first
$r\kappa$ sellers. Each winning buyer receives $\Delta_B$ and each winning
seller sells $\Delta_S$.

If $\alpha_\kappa>\hat b_{(\kappa)}$, perform trade reduction. Set
$p^\ast=\hat b_{(\kappa)}$, trade only the first $\kappa-1$ buyers, and let
\[
  T_S=\{j\in \cup_{l\ge l_S}J_l:\hat s_j\le \hat b_{(\kappa)}\}.
\]
The sellers in $T_S$ split the required quantity equally, so each sells
$(\kappa-1)\Delta_B/|T_S|$. The case $\kappa=1$ gives zero trade.

Suppose next that $\Delta_B<\Delta_S$ and set $r=\Delta_S/\Delta_B$. One
seller and $r$ buyers form one candidate trade set. The effective buyer
bid of the
$h$th trade set is
\[
  c_h=\hat b_{(r(h-1)+1)}.
\]
Let
\[
  \kappa=
  \max\bigl(\{0\}\cup
  \{h:1\le h\le m,\ r(h-1)+1\le n,\ \hat s_{(h)}\le c_h\}\bigr).
\]
If $\kappa=0$, no trade occurs. Otherwise define
\[
  \beta_\kappa=
  \begin{cases}
    \hat b_{(r\kappa+1)}, & r\kappa+1\le n,\\
    -\infty, & r\kappa+1>n.
  \end{cases}
\]
If $\beta_\kappa\ge \hat s_{(\kappa)}$, set the price
$p^\ast=\beta_\kappa$ and trade the first $r\kappa$ buyers and first $\kappa$
sellers. Each winning buyer receives $\Delta_B$ and each winning seller sells
$\Delta_S$.

If $\beta_\kappa<\hat s_{(\kappa)}$, perform trade reduction. Set
$p^\ast=\hat s_{(\kappa)}$, trade only the first $\kappa-1$ sellers, and let
\[
  T_B=\{i\in \cup_{l\ge l_B}I_l:\hat b_i\ge \hat s_{(\kappa)}\}.
\]
The buyers in $T_B$ split the required quantity equally, so each buys
$(\kappa-1)\Delta_S/|T_B|$. The case $\kappa=1$ gives zero trade.

\paragraph{$M_2$: AMM-to-buyer clearing.}
The remaining two modes use the same threshold idea but isolate one-sided
trades against the AMM\@.
Sample $l_B$ uniformly, set $\Delta=2^{l_B}\underline q$, retain buyers in
$\cup_{l\ge l_B}I_l$, and sort them by decreasing bid. If no buyer is retained,
no trade occurs. Otherwise, choose
\[
  k^\ast=\max\{k\in\{0,\ldots,n-1\}:
  k\,\hat b_{(k+1)}\Delta\ge C(k\Delta),\ k\Delta<x_0\}.
\]
If $k^\ast=0$, no trade occurs. Otherwise, the first $k^\ast$ buyers each buy
$\Delta$ units from the AMM at uniform price $p^\ast=\hat b_{(k^\ast+1)}$.

\paragraph{$M_3$: Seller-to-AMM clearing.}
Sample $l_S$ uniformly, set $\Delta=2^{l_S}\underline q$, retain sellers in
$\cup_{l\ge l_S}J_l$, and sort them by increasing ask. If no seller is retained,
no trade occurs. Otherwise, choose
\[
  k^\ast=\max\{k\in\{0,\ldots,m-1\}:
  k\,\hat s_{(k+1)}\Delta\le -C(-k\Delta)\}.
\]
If $k^\ast=0$, no trade occurs. Otherwise, the first $k^\ast$ sellers each
sell $\Delta$ units to the AMM at uniform price $p^\ast=\hat s_{(k^\ast+1)}$.

\section{Analysis}

Having described the rule, we turn to its incentive and volume guarantees.
We first verify that every realization of the mechanism is a feasible
uniform-price allocation and that truthful users satisfy ex-post IR\@. We then
state the ex-post WQIC and volume-approximation guarantees, giving the main
intuition for each; their full proofs appear in the appendix. The mechanism
neither elicits nor uses the public penalty factors, and the guarantees hold
for every vector $\boldsymbol{\lambda}$ with $\lambda_j\ge1$.

\begin{theorem}[Feasibility and ex-post IR]
\label{thm:feas-ir}
The mechanism $M$ is feasible and uniform-price in every realization. It
satisfies ex-post IR under truthful reporting.
\end{theorem}

\begin{proof}
Uniform pricing is immediate from construction. We first check feasibility.
For $M_1$, assume $\Delta_B\ge\Delta_S$; the other case is symmetric. In the
no-reduction branch, $\kappa$ buyers each receive $\Delta_B$ and $r\kappa$
sellers each sell $\Delta_S$, where $r=\Delta_B/\Delta_S$. The two sides
therefore trade the same total quantity, so the AMM state is unchanged. In the
trade-reduction branch, $\kappa-1$ buyers each receive $\Delta_B$ and sellers
in $T_S$ jointly supply $(\kappa-1)\Delta_B$, so the AMM again has zero net
exchange.

Each retained buyer has reported quantity at least $\Delta_B$. In the
no-reduction branch each retained seller has reported quantity at least
$\Delta_S$. In the reduction branch, $|T_S|\ge r(\kappa-1)$, so each selected
seller supplies at most
\[
  \frac{(\kappa-1)\Delta_B}{|T_S|}
  \le \Delta_S.
\]
Thus user quantity feasibility holds in $M_1$.

For $M_2$, if $k^\ast=0$, then no trade occurs and feasibility is trivial.
Otherwise, let $Q=k^\ast\Delta$ and $p^\ast=\hat b_{(k^\ast+1)}$. By definition,
$Qp^\ast\ge C(Q)$ and $Q<x_0$. Therefore
\[
  (x_0-Q)(y_0+Qp^\ast)
  \ge (x_0-Q)(y_0+C(Q))
  =x_0y_0.
\]
For $M_3$, if $k^\ast=0$, then no trade occurs and feasibility is trivial.
Otherwise, let $Q=k^\ast\Delta$ and $p^\ast=\hat s_{(k^\ast+1)}$. The definition
gives $Qp^\ast\le -C(-Q)$, and therefore
\[
  (x_0+Q)(y_0-Qp^\ast)
  \ge (x_0+Q)(y_0+C(-Q))
  =x_0y_0.
\]

It remains to check IR\@. In every submechanism, a truthful winning buyer trades
at a price no larger than her unit value, and a truthful winning seller trades
at a price no smaller than her unit value. Feasibility also ensures that a
truthful winner is never allocated more than her true quantity, so the
over-supply penalty is zero for truthful sellers. Each truthful realized
utility is therefore non-negative.
\hfill\qed
\end{proof}

\begin{theorem}[Ex-post WQIC]
\label{thm:wqic}
The mechanism $M$ satisfies ex-post WQIC.
\end{theorem}

This theorem is the incentive guarantee behind the mechanism. It shows that the
randomization is used only for approximation, not to average away incentive
violations: after the mode and thresholds are realized, the resulting
deterministic clearing rule is already strategically well behaved in the WQIC
sense. The proof fixes an arbitrary realization and analyzes the induced
deterministic mechanism. For the peer-to-peer mode, consider the case $\Delta_B\ge \Delta_S$;
the other case is symmetric. In the AMM-only modes, and for the buyer side of the
peer-to-peer mode, users face standard critical-price rules once their quantity
report is fixed. The only delicate case is the seller side of the
peer-to-peer mode, where trade reduction may split supply among several sellers;
there the proof uses monotonicity of the clearing index and boundary price to
show that ask deviations cannot improve utility in the cases allowed by WQIC\@.
The quantity-understatement part follows from the threshold structure: an
understatement either leaves the user in the same bucket decision and changes
nothing, or removes the user from eligibility and gives zero utility. The full
proof is in Appendix~\ref{app:wqic}.

\begin{theorem}[Volume approximation]
\label{thm:approx}
For every $\gamma\ge1$, the user-side volume $V_M$ of $M$ satisfies
\[
  \E[V_M]\ge \frac{V^\ast}{6L^2}-\frac83\bar q,
  \qquad L=\max\{1,\lceil\log_2\gamma\rceil\}.
\]
\end{theorem}

The bound is $\Omega(1/(1+\log\gamma)^2)V^\ast-O(\bar q)$ uniformly
for $\gamma\ge1$, and has the original $\Omega(1/\log^2\gamma)$ form
as $\gamma\to\infty$. If $\gamma=1$, the quantity domain is a singleton
and $L=1$. There are no quantity deviations; Theorems~\ref{thm:feas-ir}
and~\ref{thm:wqic} therefore give full IC and a constant-factor volume
guarantee up to the additive term.

This theorem is the performance guarantee behind the bucket construction. The
proof compares the mechanism to a fixed optimal feasible IR uniform-price
allocation and decomposes that optimum into the three kinds of volume that our
modes target: peer-to-peer volume, AMM-to-buyer volume, and seller-to-AMM
volume. For the AMM-only modes, a random threshold keeps a logarithmic fraction
of the optimal truncated quantity, and the constant-product average-price
monotonicity preserves AMM feasibility up to one critical-price boundary block.
For the peer-to-peer mode, the main obstacle is that buyer and seller quantities
may fall in different buckets. The proof uses a transportation decomposition
between buyer-bucket mass and seller-bucket mass and then averages over the two
independent thresholds; a threshold pair captures a constant fraction of each
bucket-pair contribution in aggregate, producing the extra logarithmic loss. We
state a key lemma for the peer-to-peer mode here; it is the local threshold
guarantee that lets the global transportation argument reduce to bounding
truncated mass on the two sides. Its proof appears in
Appendix~\ref{app:approx}.

\begin{lemma}[Threshold guarantee for $M_1$]
\label{lem:m1-threshold}
Fix a feasible IR uniform-price allocation at price $p^\ast$ and a threshold pair
$(l,t)$. Let $\Delta_B=2^l\underline q$ and $\Delta_S=2^t\underline q$. Let
$B_l$ be the retained buyers with bid at least $p^\ast$, and let $S_t$ be the
retained sellers with ask at most $p^\ast$. Then, under $M_1$ with thresholds
$(l,t)$, the one-sided matched volume is at least
\[
  \min\{|B_l|\Delta_B,\ |S_t|\Delta_S\}-2\max\{\Delta_B,\Delta_S\}.
\]
\end{lemma}

Appendix~\ref{app:approx} proves Lemma~\ref{lem:m1-threshold} and completes the
proof of Theorem~\ref{thm:approx} by aggregating this local guarantee over all
bucket pairs and combining it with the AMM-only modes.

\section{Why Full IC Is Too Strong}

The positive result achieves a polylogarithmic volume guarantee under
ex-post WQIC. The following theorem bounds worst-case volume under ex-post
full IC by an \(O(1/\gamma)\) fraction of the benchmark \(V^\ast\).
In fact, the bound holds for the larger class of mechanisms that are fully
IC in expected utility. Since ex-post full IC implies full IC in expected
utility, this establishes the claimed separation between the two ex-post
incentive requirements.

\begin{theorem}[Full IC upper bound]
\label{thm:lower}
For every randomized uniform-price mechanism $M'$ that is feasible, ex-post IR,
and fully IC in expected utility, for every $\gamma>1$, integer $N\ge 1$, reserve scale
$R>0$, reserve ratio $\rho>0$, and vector of public seller penalty factors
$\boldsymbol{\lambda}=(\lambda_1,\ldots,\lambda_N)\in[1,\infty)^N$, there is
an instance $I_{\gamma,N,R,\rho,\boldsymbol{\lambda}}$ with $N$ buyers and $N$
sellers, public penalty vector $\boldsymbol{\lambda}$, initial AMM state
$(x_0,y_0)=(R,\rho R)$, and
$\bar q/\underline q=\gamma$, such that
\[
  V^\ast(I_{\gamma,N,R,\rho,\boldsymbol{\lambda}})=2N\bar q,
  \qquad
  \E[V_{M'}(I_{\gamma,N,R,\rho,\boldsymbol{\lambda}})]\le
  \frac{3}{\gamma}V^\ast(I_{\gamma,N,R,\rho,\boldsymbol{\lambda}}).
\]
\end{theorem}

The construction applies at every prescribed initial reserve ratio $\rho$ and
allows the common reserve scale $R$ to be arbitrarily large. The user values
are rescaled with $\rho$, so this is an existential worst-case construction for
each reserve composition rather than a claim that one fixed user profile is
hard for all reserve ratios. Because
$V^\ast=2N\bar q$, the positive result's additive $O(\bar q)$ loss is
$O(1/N)$ after normalization. Our positive-result mechanism also matches the
$O(1/\gamma)$ dependence up to constants: fixing all of its bucket thresholds
to $\underline q$ gives a fully IC variant with an $\Omega(1/\gamma)$
guarantee up to the boundary loss.

Appendix~\ref{app:lower-proof} gives the complete proof of
Theorem~\ref{thm:lower}. To show that the limitation also holds when the AMM is
active in the benchmark, we give the following companion result.

\begin{proposition}[AMM-active full-IC upper bound]
\label{prop:lower-amm-active}
For every randomized uniform-price mechanism $M'$ that is feasible, ex-post IR,
and fully IC in expected utility, for every $\gamma>1$, integer $N\ge 1$, reserve scale
$R>0$, reserve ratio $\rho>0$, and vector of public seller penalty factors
$\boldsymbol{\lambda}=(\lambda_1,\ldots,\lambda_N)\in[1,\infty)^N$, there is
an instance $I^{\mathrm{AMM}}_{\gamma,N,R,\rho,\boldsymbol{\lambda}}$ with $N$
sellers and no buyers, public penalty vector $\boldsymbol{\lambda}$, initial
AMM state $(x_0,y_0)=(R,\rho R)$, and quantity bounds $\bar q=R/N$ and
$\bar q/\underline q=\gamma$, such that
\[
  V^\ast(I^{\mathrm{AMM}}_{\gamma,N,R,\rho,\boldsymbol{\lambda}})=N\bar q=R,
  \qquad
  \E[V_{M'}(I^{\mathrm{AMM}}_{\gamma,N,R,\rho,\boldsymbol{\lambda}})]\le
  \frac{3}{\gamma}
  V^\ast(I^{\mathrm{AMM}}_{\gamma,N,R,\rho,\boldsymbol{\lambda}}).
\]
The benchmark trades all $N\bar q$ units from the sellers to the AMM.
\end{proposition}

Thus the full-IC limitation is not confined to an inactive AMM. The proof of
Proposition~\ref{prop:lower-amm-active} is in
Appendix~\ref{app:lower-proof}.

\section{Conclusion}
Our results separate two ways of handling private quantities in uniform-price
AMM markets. Full two-dimensional IC limits worst-case user-side volume to an
$O(1/\gamma)$ fraction of the feasible optimum. Binary quantity thresholds,
combined with per-realization WQIC, recover a polylogarithmic guarantee
up to an additive boundary loss while
preserving uniform pricing, individual rationality, and the pool's reserve
constraint. The upper bound applies at every initial reserve ratio and at
arbitrary reserve scale, even when optimal volume is large relative to any
individual order. Thus adding liquidity alone does not remove this incentive
obstruction. The AMM-active construction further shows that the obstruction
can persist when the benchmark trades entirely against the pool.

The constructive proof identifies the role of quantity heterogeneity.
Thresholds replace heterogeneous reported quantities by fixed blocks, making
critical-price reasoning possible; independently selecting a threshold on
each side costs $L^2$ in peer-to-peer volume. Coordinating these thresholds
may offer sharper guarantees. Whether ex-post WQIC admits an
$\Omega(1/\log\gamma)$ or constant-factor guarantee, or instead requires
some logarithmic loss, remains open.

The threshold design links the incentive and volume arguments. Once a mode
and its thresholds are fixed, retained orders trade in blocks whose sizes
are independent of the exact quantity reports above the thresholds.
Critical prices and the sharing rule then establish WQIC for each realized
clearing rule. Randomizing over these rules recovers volume across quantity
scales while preserving their ex-post guarantees. Together with the upper bound for full IC in expected utility, this
establishes a separation in worst-case volume between ex-post WQIC and
ex-post full IC.

Another direction is to study the symmetric relaxation in the value dimension.
Our WQIC notion conditionally protects unit-value reports at each fixed
quantity report and rules out quantity understatements at the truthful value. One
could instead ask for mechanisms that are fully IC in
quantities but satisfy only a weak form of value IC, for
example ruling out only value understatements or overstatements that preserve
non-negative realized utility. Such a ``weak value IC'' relaxation may lead to a
different tradeoff between volume, uniform pricing, and AMM feasibility.

Finally, the objective measures executed user order flow. Our approximation
compares a mechanism and the feasible optimum on the same input; it is not
a welfare or equilibrium-volume guarantee. A welfare analysis would need to
account for surplus margins and the value of the AMM's ending inventory.
Other reserve functions, fees, multiple pools, and dynamic batches offer
further directions, but also change the feasibility and pricing arguments.
The central question in these extensions is which strategic protections can
coexist with substantial order execution at a uniform price.


\clearpage
\bibliographystyle{splncs04}
\bibliography{refs}

\clearpage
\appendix

\section{A Concrete Example for $M_1$}
\label{app:m1-example}

Consider thresholds $\Delta_B=4$ and $\Delta_S=2$, so $r=2$: one buyer block
is paired with two seller blocks. Every retained buyer reports quantity at
least $4$, and every retained seller reports quantity at least $2$. Listed in
the mechanism's sorted order, suppose the retained bids are
\[
  10>8>6
\]
and the retained asks are
\[
  3<4<7<9<11.
\]
The effective seller asks are $a_1=3$, $a_2=7$, and $a_3=11$. Hence the first two
candidate trade sets satisfy the crossing condition ($3\le10$ and $7\le8$),
while the third does not ($11>6$), so $\kappa=2$. The boundary ask is
$\alpha_\kappa=\hat s_{(5)}=11>\hat b_{(2)}=8$, and the mechanism therefore
uses trade reduction. It sets the uniform price to $8$, trades only the first
$\kappa-1=1$ buyer for four units, and forms
\[
  T_S=\{j:\hat s_j\le8\}=\{1,2,3\}.
\]
These three sellers split the four units equally, each selling $4/3\le
\Delta_S\le\hat q_j$. The winning buyer receives
$\Delta_B\le\hat q_i$. Total buyer and seller allocations are both four, so
the AMM is unchanged; every winner accepts the uniform price; and the marginal
buyer, who sets that price, is removed.

For comparison, replace the ask profile by $3<4<7<7.2<7.5$. The clearing index
is still $\kappa=2$, but now $\alpha_\kappa=7.5\le8$. The no-reduction branch
trades the first two buyers for four units each and the first four sellers for
two units each, all at price $7.5$. This illustrates how the single boundary
comparison selects between the no-reduction and trade-reduction branches.

\section{Proof of Theorem~\ref{thm:wqic}}
\label{app:wqic}

\begin{proof}
Fix a realization $\omega$ of the randomization. Thus the selected mode and
the sampled threshold or thresholds are fixed. We prove the two WQIC
requirements for the deterministic mechanism $M^\omega$.

Throughout the proof, all ties are resolved by the fixed report-independent
priority rule. In every critical-price statement below, equality with the
numerical threshold is subject to this rule; a tied user need not win. For $M_1$, we prove the case $\Delta_B\ge\Delta_S$ and set
$r=\Delta_B/\Delta_S$; the case $\Delta_B<\Delta_S$ is symmetric after
exchanging the roles of buyers and sellers.

\medskip
\noindent
\emph{Critical-price fact.}
Suppose a buyer with true type $(b,q)$ receives a fixed block $Q$ iff her bid
is at least a threshold $\tau$ that is independent of her own bid, receives
zero otherwise, and pays $\tau Q$ whenever she receives $Q$. Then every winning
bid gives the same utility
\[
  b\min\{Q,q\}-\tau Q,
\]
and every losing bid gives utility zero. Hence truthful bidding is optimal
among all bid reports whenever truthful bidding gives non-negative realized
utility; moreover, if truthful bidding gives non-positive realized utility,
then every bid report gives non-positive realized utility. The seller analogue
holds when a seller sells a fixed block $Q$ iff her ask is at most a threshold
$\tau$ that is independent of her own ask, receives zero otherwise, and is paid
$\tau Q$ whenever she sells.

Table~\ref{tab:wqic-critical} records the cases in which the critical-price
fact applies directly.

\begin{table}[ht]
\centering
\small
\caption{Direct critical-price cases; equality is resolved by the fixed priority rule.}
\label{tab:wqic-critical}
\begin{tabular}{c c c c c}
\hline
Mode & Side & Winning set & Block & Critical price \\
\hline
$M_2$ & Buyer &
$\hat b\ge \tau_i$ &
$\Delta$ &
$\tau_i=\hat b^{-i}_{(K_i)}$ \\

$M_3$ & Seller &
$\hat s\le \tau_j$ &
$\Delta$ &
$\tau_j=\hat s^{-j}_{(K_j)}$ \\

$M_1$ & Buyer &
$\hat b\ge \tau_i$ &
$\Delta_B$ &
$\tau_i=\min\{\hat s_{(rK_i+1)},\hat b^{-i}_{(K_i-1)}\}$ \\
\hline
\end{tabular}
\end{table}

For $M_2$, fix the sampled block size $\Delta$ and all other retained buyer
reports. Let buyer $i$ be the distinguished retained buyer, and write the other
retained buyer bids as
\[
  \hat b^{-i}_{(1)}\ge\hat b^{-i}_{(2)}\ge\cdots.
\]
If buyer $i$ is among the first $k$ buyers in the full order, then the first
losing bid is $\hat b^{-i}_{(k)}$. Hence the AMM feasibility test for clearing
$k$ buyers is
\[
  k\hat b^{-i}_{(k)}\Delta\ge C(k\Delta),
  \qquad k\Delta<x_0,
\]
which is independent of buyer $i$'s own bid. Let
\[
  K_i=\max\{k\ge 1:
      k\hat b^{-i}_{(k)}\Delta\ge C(k\Delta),\ k\Delta<x_0
   \},
\]
where $k$ ranges over available other buyers, with $K_i=0$ if the set
is empty. The bids decrease with $k$, while $C(k\Delta)/(k\Delta)$
increases, so feasible positive indices form an initial segment.
If $K_i=0$, buyer $i$ never wins. If
$K_i>0$, buyer $i$ wins iff her bid is at least
$\hat b^{-i}_{(K_i)}$, receives $\Delta$, and pays
$\hat b^{-i}_{(K_i)}\Delta$. Thus buyer $i$ faces a critical bid
$\tau_i=\hat b^{-i}_{(K_i)}$.

The argument for $M_3$ is symmetric. Fix the sampled block size $\Delta$ and
all other retained seller reports. Let seller $j$ be the distinguished retained
seller, and write the other retained seller asks as
\[
  \hat s^{-j}_{(1)}\le\hat s^{-j}_{(2)}\le\cdots.
\]
If seller $j$ is among the first $k$ sellers in the full order, then the first
losing ask is $\hat s^{-j}_{(k)}$, so the AMM feasibility test for clearing
$k$ sellers is
\[
  k\hat s^{-j}_{(k)}\Delta\le -C(-k\Delta),
\]
which is independent of seller $j$'s own ask. Let
\[
  K_j=\max\{k\ge 1:
      k\hat s^{-j}_{(k)}\Delta\le -C(-k\Delta)
   \},
\]
where $k$ ranges over available other sellers, with $K_j=0$ if the set
is empty. The asks increase with $k$, while $-C(-k\Delta)/(k\Delta)$
decreases, so feasible positive indices form an initial segment.
If $K_j=0$, seller $j$ never wins. If
$K_j>0$, seller $j$ faces a
critical ask $\tau_j=\hat s^{-j}_{(K_j)}$, sells $\Delta$ iff her ask is at
most $\tau_j$, and is paid $\tau_j\Delta$. This proves the unit-value WQIC
condition for buyers in $M_2$ and sellers in $M_3$, including the case in which
truthful unit-value reporting gives non-positive utility. The inactive side in
each mode receives zero for all unit-value reports.

We now handle the buyer side of $M_1$. Fix all seller reports and all other
buyer reports. Let buyer $i$ be the distinguished retained buyer. Write the
other retained buyer bids as
\[
  \hat b^{-i}_{(1)}\ge\hat b^{-i}_{(2)}\ge\cdots ,
\]
with the convention $\hat b^{-i}_{(0)}=+\infty$, and write the retained seller
asks as
\[
  \hat s_{(1)}\le\hat s_{(2)}\le\cdots\le\hat s_{(m)}.
\]
Define
\[
  K_i=
  \max\Bigl(
  \{0\}\cup
  \{1\le h\le n:
       r(h-1)+1\le m,\
       \hat s_{(r(h-1)+1)}\le \hat b^{-i}_{(h-1)}
   \}
  \Bigr).
\]
Here $n$ is the total number of retained buyers. If $K_i=0$, buyer $i$ is never allocated. If $K_i>0$, define
\[
  \tau_i=
  \min\{\hat s_{(rK_i+1)},\hat b^{-i}_{(K_i-1)}\},
\]
where $\hat s_{(rK_i+1)}=+\infty$ if $rK_i+1>m$.

If $\tau_i=+\infty$, buyer $i$ never wins. Otherwise,
we claim that buyer $i$ wins iff her bid is at least $\tau_i$, and whenever she
wins she receives $\Delta_B$ and pays $\tau_i\Delta_B$. If
$\hat s_{(rK_i+1)}\le \hat b^{-i}_{(K_i-1)}$, then the no-reduction branch
applies at index $K_i$: the first losing seller ask is
$\hat s_{(rK_i+1)}$, and the price is
$\hat s_{(rK_i+1)}=\tau_i$. If instead
$\hat s_{(rK_i+1)}>\hat b^{-i}_{(K_i-1)}$, then the trade-reduction branch
applies, the marginal buyer bid is $\hat b^{-i}_{(K_i-1)}$, and the price is
$\hat b^{-i}_{(K_i-1)}=\tau_i$. In either case, buyer $i$ receives the fixed
block $\Delta_B$ exactly when her bid reaches $\tau_i$, and pays
$\tau_i\Delta_B$ whenever she wins. The critical-value fact proves the
unit-value WQIC condition for buyers in $M_1$, including the case in which
truthful bidding gives non-positive utility.

It remains to handle the seller side of $M_1$, where the trade-reduction branch
may allocate a shared fractional amount. Fix all buyer reports and all seller
reports except seller $j$'s ask. Let seller $j$ have true type $(s,q)$ and
public factor $\lambda$. For an ask report $\hat s$ of seller $j$, let $\kappa(\hat s)$ be the
clearing index, $p^\ast(\hat s)$ the resulting uniform price, and $Q_j(\hat s)$ seller
$j$'s allocated quantity. Also write
\[
  N_{\hat s}(p)=|\{\text{retained sellers with ask at most }p\}|.
\]
Changing seller $j$'s ask changes each count $N_{\hat s}(p)$ by at most one. Lowering
the ask can only weakly increase the set of feasible indices, while raising the
ask can only weakly decrease it.

We first record two boundary facts:
\[
\begin{array}{ll}
\text{(B1)} &
a<s,\ \kappa(a)>\kappa(s)
\quad\Longrightarrow\quad
p^\ast(a)<s,
\\[0.4em]
\text{(B2)} &
a>s,\ \kappa(a)<\kappa(s)
\quad\Longrightarrow\quad
Q_j(a)=0.
\end{array}
\]
For (B1), let $h=\kappa(a)$. Since index $h$ becomes feasible after the
lower-ask deviation but was not feasible under truthful asking, seller $j$ must
be newly counted at a buyer boundary price $\hat b_{(h)}$ at which she was not
counted truthfully. Hence
\[
  a\le \hat b_{(h)}<s.
\]
The new uniform price is at most $\hat b_{(h)}$: in the trade-reduction branch
it is exactly~$\hat b_{(h)}$, while in the no-reduction branch it is the first
losing seller-side boundary ask and is no larger than $\hat b_{(h)}$. Therefore
$p^\ast(a)<s$.

For (B2), let $h=\kappa(a)$. If $h=0$, no trade occurs. Otherwise, since index $h+1$ was feasible under truthful
asking but is not feasible after the higher-ask deviation, seller $j$ must have
been counted at the buyer boundary price $\hat b_{(h+1)}$ under truthful asking
and not counted there after the deviation. Thus
\[
  s\le \hat b_{(h+1)}<a.
\]
After the deviation, exactly $rh$ other sellers have asks at most
$\hat b_{(h+1)}<a$, and hence at least $rh$ sellers are ranked ahead
of seller $j$. If the new outcome at index $h$ is a no-reduction outcome, these
sellers fill all seller slots, so seller $j$ is not traded. If the new outcome
is a trade-reduction outcome, seller $j$ could be in the reduction set only if
$a\le p^\ast(a)=\hat b_{(h)}$. But then the same $rh$ sellers counted at
$\hat b_{(h+1)}$, together with seller $j$, would imply
\[
  N_a(\hat b_{(h)})\ge rh+1,
\]
contradicting the fact that the mechanism entered the reduction branch at index~$h$. Hence $Q_j(a)=0$.

Table~\ref{tab:wqic-m1-seller} summarizes the deviations. In its first
four rows, assume $Q_j(s)>0$. Write $U(s)$ for truthful realized utility
and $U(a)$ for utility after the deviation. The quantity report is fixed
and may exceed the true quantity, so $U(s)$ need not be non-negative.

\begin{table}[ht]
\centering
\small
\caption{Seller-side deviations in $M_1$ (first four rows: truthful winner).}
\label{tab:wqic-m1-seller}
\begin{tabular}{c c c c}
\hline
Case & Effect on $\kappa$ & Key fact & Conclusion \\
\hline
$a<s$ & unchanged & same price and share & $U(a)=U(s)$ \\
$a<s$ & increases & (B1): $p^\ast(a)<s$ & $U(a)\le0$ \\
$a>s$ & unchanged & same outcome or exit & $U(a)\in\{U(s),0\}$ \\
$a>s$ & decreases & (B2): no allocation & $U(a)=0$ \\
$Q_j(s)=0,\ a>s$ & -- & remains unallocated & $U(a)=0$ \\
$Q_j(s)=0,\ a<s$ & -- & new trade at $p^\ast(a)\le s$ & $U(a)\le0$ \\
\hline
\end{tabular}
\end{table}

Suppose first that seller $j$ trades truthfully. Lowering her ask while
keeping the index unchanged leaves her price and quantity unchanged.
In the no-reduction branch, the price is an order statistic outside the
traded set; lowering a winning ask does not change this boundary. In the
reduction branch, the price is the fixed marginal buyer bid. The seller
already qualifies at this price, so the qualifying set and every equal
share remain unchanged. A lower ask that increases the index instead
satisfies (B1), so any positive allocation gives non-positive utility.
This proves the first two rows. Crucially, the second row asserts
$U(a)\le0$, not $U(a)\le U(s)$ when $U(s)<0$.

For a higher ask at an unchanged index, if the seller continues trading,
the same outside boundary or marginal buyer bid sets the price. She remains
within the selected set or the same reduction set, respectively, so her
quantity is unchanged too. If she leaves this set she gets zero. A higher
ask that reduces the index gives zero by (B2). Thus every deviation by a
truthful winner yields either $U(s)$ or non-positive utility.

Suppose now that seller $j$ does not trade truthfully. Raising her ask
cannot select her: at an unchanged index she remains outside the winning
set, and at a smaller index (B2) applies. If lowering her ask increases
the index, (B1) applies. At an unchanged index, entry into a reduction
outcome requires crossing the buyer price, which is at most $s$. Entry
into a no-reduction outcome places her among the first $r\kappa$ sellers;
the displaced seller determines the new outside boundary, at most her
original ask $s$. Thus any new allocation is paid at most $s$, and its
utility is non-positive. Equality at a tie causes no difficulty; all
allocation ties use the fixed priority rule.

Consequently, if $U(s)\ge0$, no deviation yields more than $U(s)$; if
$U(s)\le0$, every deviation gives non-positive utility. This establishes
both sign clauses without assuming that the fixed quantity report is
truthful. The over-supply penalty only reduces utility in the cases paid
at or below the true unit cost.

For $\Delta_B<\Delta_S$, the fixed-block side is the seller, and the
buyer side shares quantity in the reduction branch. Exchange the two
sides in the boundary argument: a higher buyer bid that increases the
index pays above true value, and a lower bid that decreases it receives
zero. At an unchanged index, a continuing winner has the same price and
quantity; a truthful loser can enter only at a price at least her true
value. Since units beyond a buyer's true demand have zero marginal value,
these newly obtained allocations have non-positive utility. The same two
sign clauses follow for buyers; fixed-block seller incentives follow from
the critical-price fact.

Combining Table~\ref{tab:wqic-critical} with
Table~\ref{tab:wqic-m1-seller}, truthful unit-value reporting satisfies the
first WQIC requirement for every user in the fixed realization $\omega$: if
truthful unit-value reporting with the fixed quantity report gives
non-negative utility, no unit-value deviation gives higher utility; and if it
gives non-positive utility, then every unit-value report with the same quantity
report gives non-positive utility.

We now prove the quantity-understatement requirement. Fix a user, her truthful
unit value, her true quantity $q$, and all other reports. Compare truthful
quantity reporting with an understated report $\hat q\le q$.

In the fixed realization $\omega$, let $\Delta$ be the sampled threshold block
size relevant for this user: for $M_2$ or $M_3$ this is the sampled block size,
and for $M_1$ it is $\Delta_B$ for a buyer and $\Delta_S$ for a seller. If both
$q$ and $\hat q$ are below $\Delta$, the user is not retained under either
report and receives zero utility in both cases. If both reports are at least
$\Delta$, then the user is retained under both reports. Conditional on
retention, the mechanism uses only the fixed block size and the unit-value
ordering; the exact reported quantity above the threshold does not affect the
clearing index, branch, price, allocation, or payment. Hence the two reports
give identical utility.

The only remaining case is
\[
  q\ge \Delta>\hat q.
\]
Then the understated report removes the user from the retained set and gives
zero utility. Under truthful quantity reporting, feasibility gives an
allocation of at most the true quantity $q$. Since the unit value is truthful,
a truthful winning buyer pays no more than her value per unit, and a truthful
winning seller is paid no less than her value per unit and incurs no
over-supply penalty. Thus truthful quantity reporting gives non-negative
realized utility, while the understated report gives zero utility. Both WQIC
requirements hold for the arbitrary realization $\omega$, so $M$ satisfies
ex-post WQIC.
\hfill\qed
\end{proof}

\section{Proof of Theorem~\ref{thm:approx}}
\label{app:approx}

We first prove Lemma~\ref{lem:m1-threshold}, the local threshold guarantee for
the peer-to-peer submechanism $M_1$, and then use it to prove
Theorem~\ref{thm:approx}.

\begin{proof}[of Lemma~\ref{lem:m1-threshold}]
Write $N_B=|B_l|$ and $N_S=|S_t|$. Suppose first that
$\Delta_B\ge \Delta_S$ and set $r=\Delta_B/\Delta_S$. For every
\[
  h\le \min\{N_B,\lfloor N_S/r\rfloor\},
\]
the $h$th retained buyer has bid at least $p^\ast$, and the seller at position
$r(h-1)+1$ has ask at most $p^\ast$. Therefore
\[
  a_h=\hat s_{(r(h-1)+1)}\le p^\ast\le \hat b_{(h)},
\]
so the mechanism's clearing index satisfies
\[
  \kappa\ge \min\{N_B,\lfloor N_S/r\rfloor\}.
\]
Hence $M_1$ trades at least
\[
  \min\{N_B,\lfloor N_S/r\rfloor\}-1
\]
buyer blocks. Its one-sided matched volume is therefore at least
\[
  \Bigl(\min\{N_B,\lfloor N_S/r\rfloor\}-1\Bigr)\Delta_B.
\]
Since $\lfloor N_S/r\rfloor\Delta_B=\lfloor N_S/r\rfloor r\Delta_S\ge
N_S\Delta_S-\Delta_B$, this lower bound is at least
\[
  \min\{N_B\Delta_B,\ N_S\Delta_S\}-2\Delta_B
  =\min\{N_B\Delta_B,\ N_S\Delta_S\}-2\max\{\Delta_B,\Delta_S\}.
\]

Now suppose that $\Delta_B<\Delta_S$ and set $r=\Delta_S/\Delta_B$. For every
$h\le \min\{N_S,\lfloor N_B/r\rfloor\}$, the $h$th retained seller has ask at
most $p^\ast$, and the buyer at position $r(h-1)+1$ has bid at least
$p^\ast$. Thus $\hat s_{(h)}\le p^\ast\le
c_h=\hat b_{(r(h-1)+1)}$, so
\[
  \kappa\ge \min\{N_S,\lfloor N_B/r\rfloor\}.
\]
Hence $M_1$ trades at least $\min\{N_S,\lfloor N_B/r\rfloor\}-1$ seller blocks,
and its one-sided matched volume is at least
\[
  \Bigl(\min\{N_S,\lfloor N_B/r\rfloor\}-1\Bigr)\Delta_S.
\]
Because $\lfloor N_B/r\rfloor\Delta_S\ge N_B\Delta_B-\Delta_S$, this is at
least
\[
  \min\{N_S\Delta_S,\ N_B\Delta_B\}-2\Delta_S
  =\min\{N_B\Delta_B,\ N_S\Delta_S\}-2\max\{\Delta_B,\Delta_S\}.
\]
This proves the lemma.
\hfill\qed
\end{proof}

\begin{proof}[of Theorem~\ref{thm:approx}]
Fix an optimal feasible IR allocation at price $p$, with total buyer and
seller quantities $Q_B,Q_S$. Put $Q=\min\{Q_B,Q_S\}$,
$A_B=(Q_B-Q_S)_+$, and $A_S=(Q_S-Q_B)_+$, where $(a)_+=\max\{a,0\}$.
Decompose actual allocated quantities into a peer-to-peer component of
one-sided volume $Q$ and residual AMM components $A_B,A_S$. This can split
an individual allocation and requires no removal of partially served users.
In particular,
\[
 V^*=2Q+A_B+A_S.
\]
Every served buyer has value at least $p$ and every served seller has value
at most $p$. Let $n_k,m_k$ count these buyers and sellers in bucket $k$,
and define their retained truncated capacities
\[
 x_l=\sum_{k\ge l}n_k2^l\underline q,\qquad
 y_t=\sum_{k\ge t}m_k2^t\underline q.
\]

\emph{AMM components.}
For $A_B>0$, feasibility of the optimal allocation implies
$pA_B\ge C(A_B)$ and $A_B<x_0$. Average-price monotonicity implies
$pq\ge C(q)$ for $0\le q\le A_B$. For threshold block
$\Delta=2^l\underline q$, the integer
\[
 k=\min\{\lfloor A_B/\Delta\rfloor,\sum_{h\ge l}n_h-1\}
\]
is feasible for $M_2$ whenever $k\ge1$: its losing bid is at least $p$,
and $k\Delta\le A_B$. Consequently, including the case $k\le0$,
\begin{equation}\label{eq:amm-cap}
 V_{M_2}(l)\ge\min\{x_l,A_B\}-2\Delta.
\end{equation}
The cap by $A_B$ is essential; $x_l$ alone need not be AMM-feasible.
Let $w_k$ be the residual AMM allocation to buyers in bucket $k$.
Then $\sum_k w_k=A_B$ and $w_k\le n_k2^{k+1}\underline q$, including
the closed last bucket. Therefore
\[
 x_l\ge a_l:=\sum_{k\ge l}w_k2^{l-k-1},\qquad a_l\le A_B,
 \qquad \sum_l a_l=\sum_k w_k(1-2^{-k-1})\ge A_B/2.
\]
Thus $\sum_l\min\{x_l,A_B\}\ge A_B/2$, and averaging
\eqref{eq:amm-cap} gives
$\E[V_{M_2}]\ge A_B/(2L)-2\bar q$. If $A_B=0$ this inequality is
trivial. For $M_3$, feasibility gives $pA_S\le-C(-A_S)$, and
the average buy price is decreasing. Repeating the argument with sellers
gives $\E[V_{M_3}]\ge A_S/(2L)-2\bar q$.

\emph{Peer-to-peer component.}
Let $u_k,v_r$ be its buyer and seller bucket masses, both summing to $Q$.
Choose a non-negative transportation matrix $(q_{kr})$ with these row
and column sums. Since $u_k\le n_k2^{k+1}\underline q$ and
$v_r\le m_r2^{r+1}\underline q$, for every $(l,t)$,
\[
 x_l\ge\sum_{k\ge l,r\ge t}q_{kr}2^{l-k-1},\qquad
 y_t\ge\sum_{k\ge l,r\ge t}q_{kr}2^{t-r-1}.
\]
Using $\min\{\sum_e a_e,\sum_e b_e\}\ge\sum_e\min\{a_e,b_e\}$
for non-negative families and summing over thresholds,
\begin{align*}
 \sum_{l,t}\min\{x_l,y_t\}
 &\ge\frac12\sum_{k,r}q_{kr}
       \sum_{l=0}^k\sum_{t=0}^r\min\{2^{l-k},2^{t-r}\}\\
 &\ge\frac12\sum_{k,r}q_{kr}=Q/2.
\end{align*}
The last inequality uses the term $l=k,t=r$. By
Lemma~\ref{lem:m1-threshold}, the user-side volume at $(l,t)$ is at least
$2\min\{x_l,y_t\}-4\bar q$. Hence
$\E[V_{M_1}]\ge Q/L^2-4\bar q$.
Finally, each mode is selected with probability $1/3$, so
\[
 \E[V_M]\ge\frac{Q}{3L^2}+\frac{A_B+A_S}{6L}-\frac83\bar q
 \ge\frac{V^*}{6L^2}-\frac83\bar q.
\]
\end{proof}

\paragraph{The fixed-threshold fully IC variant.}
For completeness, fix every threshold at $\delta=\underline q$ and keep
uniform random selection of the three modes. Every permitted quantity
report is retained and has no effect on the outcome. Every allocation is
at most $\delta$, hence at most the true quantity. The value part of WQIC
and IR therefore give dominant truthful value reporting for every
realization, and quantity reports are immaterial. This proves full IC,
including joint deviations and expected-utility IC.

Use the same optimal-allocation decomposition as above, and let $n,m$
count its served buyers and sellers. Since each contributes at most
$\bar q$, we have $n\delta,m\delta\ge Q/\gamma$ and
$n\delta\ge A_B/\gamma$, $m\delta\ge A_S/\gamma$.
Lemma~\ref{lem:m1-threshold} gives peer user volume at least
$2Q/\gamma-4\delta$. The capped one-sided argument gives at least
$A_B/\gamma-2\delta$ in $M_2$ and $A_S/\gamma-2\delta$ in $M_3$,
because $\min\{n\delta,A_B\}\ge A_B/\gamma$ and symmetrically for
sellers. Averaging the modes yields
\[
 \E[V_{\mathrm{fixed}}]\ge \frac{V^*}{3\gamma}-\frac83\underline q.
\]
This supplies the claimed matching dependence on $\gamma$, up to constants
and the boundary loss, without changing the randomized-threshold mechanism.

\section{Full-IC Upper Bounds}
\label{app:lower-proof}

We first prove Theorem~\ref{thm:lower} and then prove its complementary
AMM-active version, Proposition~\ref{prop:lower-amm-active}.

\begin{proof}[Theorem~\ref{thm:lower}]
Fix an arbitrary vector of public penalty factors
$\boldsymbol{\lambda}=(\lambda_1,\ldots,\lambda_N)\in[1,\infty)^N$. Set
\[
  b_1=\rho,\qquad s_1=\frac{\rho}{2},\qquad
  s_{0j}=\frac{\rho}{6\lambda_j}\quad\text{for each seller }j,
\]
so $s_1-\lambda_j s_{0j}=\rho/3$ for every $j$. Let
$q_0=\underline q$ and
$q_1=\bar q=\gamma q_0$. All buyers in the profiles below have high type
$\theta_B^1=(b_1,q_1)$, and all sellers have high type
$\theta_S^1=(s_1,q_1)$. For each seller $j$, define the auxiliary low type
$\theta_{S,j}^0=(s_{0j},q_0)$. For the prescribed $R>0$ and $\rho>0$, choose
AMM reserves $x_0=R$ and $y_0=b_1R=\rho R$.
These auxiliary low types may differ across sellers, but each is used only in
the corresponding profile $L_j$ below; the hard profile $H$ retains the same
high type for every seller.

Let $H$ be the all-high profile with $N$ high buyers and $N$ high sellers; this
will be the instance $I_{\gamma,N,R,\rho,\boldsymbol{\lambda}}$. In $H$, the
benchmark can match every buyer and every seller peer-to-peer for $q_1$ units
at any price in $[s_1,b_1]$, leaving the AMM unchanged. Hence
\[
  V^\ast(H)=2Nq_1=2N\bar q.
\]

Fix a seller $j$. Let $L_j$ be the profile obtained from $H$ by changing only
seller $j$ to her auxiliary low type $\theta_{S,j}^0$. Write $z_j^H,t_j^H$
for seller $j$'s random allocation and receipt in $H$, and $z_j^L,t_j^L$ for
the corresponding random variables when the low seller reports truthfully in
$L_j$.

First bound the low seller's truthful utility
\[
  U_j^L=\E[t_j^L-s_{0j}z_j^L].
\]
Since seller $j$ reports capacity $q_0$ in $L_j$, feasibility gives
$z_j^L\le q_0$ in every realization. If some buyer receives a positive
allocation in that realization, buyer realized IR and value $b_1$ imply
$\pi\le b_1$, and so $t_j^L=\pi z_j^L\le b_1q_0$. If no buyer receives a
positive allocation but seller $j$ does, the AMM net buys
$D=Q_S-Q_B=Q_S>0$. LND feasibility gives
\[
  (x_0+D)(y_0-\pi D)\ge x_0y_0,
\]
and hence
\[
  \pi\le \frac{y_0}{x_0+D}<\frac{y_0}{x_0}=b_1.
\]
Thus, realization by realization,
\[
  t_j^L-s_{0j}z_j^L\le b_1q_0,
  \qquad\text{and hence}\qquad
  U_j^L\le b_1q_0.
\]

Full IC for seller $j$ in $L_j$ says that the low seller cannot gain by
reporting $\theta_S^1$. Under that deviation the reported profile is exactly
$H$, while the true type remains $(s_{0j},q_0)$ and the public factor remains
$\lambda_j$. Therefore
\[
\begin{aligned}
  U_j^L
  &\ge
  \E\!\left[
    t_j^H-s_{0j}z_j^H
    -(\lambda_j-1)s_{0j}\max\{z_j^H-q_0,0\}
  \right]  \\
  &\ge \E[t_j^H-\lambda_j s_{0j}z_j^H].
\end{aligned}
\]
In $H$, seller $j$ is a truthful high seller. Realized IR gives
$t_j^H-s_1z_j^H\ge 0$ in every realization. Combining the last two displays,
\[
  b_1q_0
  \ge \E[t_j^H-\lambda_j s_{0j}z_j^H]
  \ge (s_1-\lambda_j s_{0j})\E[z_j^H].
\]
Since $s_1-\lambda_j s_{0j}=\rho/3$ and $b_1=\rho$,
\[
  \E[z_j^H]\le 3q_0.
\]
Summing over all sellers gives
\[
  \E\!\left[\sum_{j=1}^N z_j^H\right]\le 3Nq_0.
\]

In each realization of $H$, let $Q_B=\sum_i z_i^H$ and
$Q_S=\sum_j z_j^H$. Suppose, toward a contradiction, that $Q_B>Q_S$, and set
$D=Q_B-Q_S>0$. Some buyer then receives a positive allocation, so ex-post IR
implies $\pi\le b_1$. On the other hand, using $y_0=b_1x_0$, LND feasibility
requires
\[
  (x_0-D)(b_1x_0+\pi D)\ge b_1x_0^2.
\]
After expansion, this is
\[
  D\bigl((\pi-b_1)x_0-\pi D\bigr)\ge0.
\]
But reserve non-negativity and $\pi\le b_1$ imply
\[
  (\pi-b_1)x_0-\pi D
  =\pi(x_0-D)-b_1x_0
  \le -b_1D<0,
\]
a contradiction. Therefore $Q_B\le Q_S$ in every realization. Taking
expectations,
\[
\begin{aligned}
  \E[V_{M'}(H)]
  &=\E[Q_B+Q_S] \\
  &\le 2\E[Q_S]
   \le 6Nq_0.
\end{aligned}
\]
Dividing by $V^\ast(H)=2Nq_1=2N\gamma q_0$ gives
\[
  \frac{\E[V_{M'}(H)]}{V^\ast(H)}
  \le \frac{3}{\gamma}.
\]
The prescribed $R$ and $\rho$ were arbitrary, so the bound holds at every
reserve scale and every prescribed reserve ratio. The construction sets the
hard-instance buyers' common unit value equal to $y_0/x_0=\rho$; IR then
prevents a positive net AMM sale to buyers. The user profile therefore depends
on the prescribed ratio.
\hfill\qed
\end{proof}

\begin{proof}[Proposition~\ref{prop:lower-amm-active}]
Fix an arbitrary vector of public penalty factors
$\boldsymbol{\lambda}=(\lambda_1,\ldots,\lambda_N)\in[1,\infty)^N$. Let
\[
  q_1=\bar q=\frac{R}{N},\qquad
  q_0=\underline q=\frac{q_1}{\gamma},\qquad
  s_1=\frac{\rho}{2},\qquad
  s_{0j}=\frac{\rho}{6\lambda_j}\quad\text{for each seller }j.
\]
These seller-specific low types are used only in the auxiliary profiles
$L_j$ below; all sellers in the hard profile $H$ have the same high type.
There are no buyers. Let
$H=I^{\mathrm{AMM}}_{\gamma,N,R,\rho,\boldsymbol{\lambda}}$ be the profile of $N$
sellers, each with high type $(s_1,q_1)$ and public penalty factor $\lambda_j$,
and let the initial AMM state be $(x_0,y_0)=(R,\rho R)$.
At the uniform price $\pi=\rho/2$, all sellers can sell $q_1$ units to the AMM.
The total sale is $Nq_1=R$, and the final AMM state is
$(2R,\rho R/2)$, whose reserve product equals the initial product. Thus this
allocation is feasible and IR. Seller capacities also give the matching upper bound, so
\[
  V^\ast(H)=Nq_1=R.
\]

Fix a seller $j$, and let $L_j$ be the profile obtained by changing only this
seller to her auxiliary low type $(s_{0j},q_0)$. Under truthful reporting in
$L_j$, feasibility gives $z_j^L\le q_0$. If seller $j$ trades in a
realization, the AMM's total net purchase $Q_S$ is positive, and LND
feasibility implies
\[
  (R+Q_S)(\rho R-\pi Q_S)\ge\rho R^2,
  \qquad\text{hence}\qquad
  \pi\le\frac{\rho R}{R+Q_S}<\rho.
\]
Consequently the low seller's expected truthful utility satisfies
$U_j^L\le\rho q_0$.

Full IC for the deviation from $(s_{0j},q_0)$ to $(s_1,q_1)$, with the public
factor $\lambda_j$ held fixed, followed by
ex-post IR for the high type in $H$, gives
\[
\begin{aligned}
  \rho q_0
  &\ge U_j^L \\
  &\ge \E\!\left[t_j^H-s_{0j}z_j^H
       -(\lambda_j-1)s_{0j}\max\{z_j^H-q_0,0\}\right] \\
  &\ge \E[t_j^H-\lambda_j s_{0j}z_j^H] \\
  &\ge (s_1-\lambda_j s_{0j})\E[z_j^H]
   =\frac{\rho}{3}\E[z_j^H].
\end{aligned}
\]
Therefore $\E[z_j^H]\le3q_0$. Summing over the $N$ sellers and using the
absence of buyers,
\[
  \E[V_{M'}(H)]\le3Nq_0
  =\frac{3}{\gamma}Nq_1
  =\frac{3}{\gamma}V^\ast(H).
\]
All benchmark volume in this construction is seller-to-AMM volume.
\hfill\qed
\end{proof}

\end{document}